%% file: finite-sample-iv-l4dc.tex
\pdfoutput=1
\documentclass[]{l4dc2026}

\usepackage{amsfonts}  
\usepackage{mathtools} 
\usepackage{commath}   
\usepackage{cancel}    
\usepackage{bm}        
\usepackage{xcolor}    
\usepackage[activate={true,nocompatibility}]{microtype} 
\usepackage{etoolbox}  
\usepackage{comment}   
\usepackage{mdframed}  
\usepackage{xspace}    
\usepackage{booktabs}  
\usepackage{siunitx}   
\usepackage{float}     
\usepackage{doi}          

\newtheorem{fact}{Fact}
\newtheorem{assum}{Assumption}
\newtheorem{problem}{Problem}
\DeclareMathOperator{\trace}{\operatorname{tr}}
\DeclareMathOperator{\expect}{\mathbb{E}}
\DeclareMathOperator{\probability}{\mathbb{P}}
\newcommand{\clip}[2][\lambda]{\ensuremath{\left[#2\right]_{\vee #1}}}

\title{Instrumental variables system identification with \(L^p\) consistency}
\usepackage{times}

\coltauthor{%
 \Name{Simon Kuang} \Email{slku@ucdavis.edu}
 \AND
 \Name{Xinfan Lin} \Email{lxflin@ucdavis.edu}\\
 \addr University of California, Davis%
}

\renewcommand{\cbr}[1]{\ensuremath{\left\{#1\right\}}}

\DeclareMathOperator{\var}{\operatorname{var}}
\DeclareMathOperator{\frob}{\operatorname{F}}

\newtoggle{preprint}
\settoggle{preprint}{true}

\newenvironment{after}{\color{red!90!black}}{\ignorespacesafterend}
\renewenvironment{after}{}{}

\begin{document}
\maketitle
\begin{abstract}
Instrumental variables (IV) eliminate the bias that afflicts least-squares identification of dynamical systems through noisy data, yet traditionally relies on external instruments that are seldom available for nonlinear time series data.
We propose an IV estimator that synthesizes instruments from the data.
We establish finite-sample \(L^{p}\) consistency for \emph{all} \(p \ge 1\) in both discrete- and continuous-time models, recovering a nonparametric \(\sqrt{n}\)-convergence rate.
On a forced Lorenz system our estimator reduces parameter bias by 200x (continuous-time) and 500x (discrete-time) relative to least squares and reduces RMSE by up to tenfold.
Because the method only assumes that the model is linear in the unknown parameters, it is broadly applicable to modern sparsity-promoting dynamics learning models.
\end{abstract}




\input{introduction-l4dc}

\input{iv-sections-l4dc/iv-analysis.tex}

\input{iv-sections-l4dc/iv-theorem.tex}

\input{iv-sections-l4dc/lorenz.tex}

\input{iv-sections-l4dc/vanderpol.tex}

\section{Novelty}
We develop a form of instrumental variables estimation for system identification problems where there are no obvious instruments.
The instruments are based on a novel application of local polynomial regression to smooth or differentiate a function onto an off-grid point.

We then show that this estimator is consistent in \(L^p\), whereas the vanilla IV estimator does not even have an expectation.

\section{Significance}
Many problems in driven engineering could benefit from the bias reduction of instrumental variables estimation.
The only shortfall is that (especially in nonlinear models) there are no instruments.
Our filtering constructions are applicable to any method that models time evolution linearly in the parameters \citep{kutz_dynamic_2016,brunton_discovering_2016,mezic_koopman_2021,haller_modeling_2025}.

The \(L^p\) consistency of our estimator appeals to a trend towards non-asymptotic analyses of point estimators traditionally understood via asymptotic normality, such as system identification with noiseless data and random designs \citep{ziemann_tutorial_2023,bakshi_new_2023}.
Broadly speaking, one hopes to attain the classical \(\sqrt{n}\)-asymptotic rate of convergence, but with non-asymptotic constants.
Whereas a large body of work focuses on bounding the quantiles of the estimation error, our work achieves this goal (with a nonparametric power-of-\(h\) caveat) by bounding the \(L^p\) risk for a fixed-design, random noise setting.
A technique that may be of independent interest is the layer cake technique for bounding regularized matrix inverses (Proposition~\ref{lem:three-way-integral}), and the \(\mu\)-truncation for thinning the tails of a design matrix from subexponential to subgaussian without incurring any bias
(an elementary workaround that could obviate more complex Hanson-Wright inequalities).



\clearpage
\acks{This work was supported by the National Science Foundation CAREER Program (Grant No. 2046292).}

\bibliography{export.bib}

\clearpage
\tableofcontents
\appendix
\clearpage
\section{Sample split illustration}
\label{section:sample-split-illustration}
\begin{figure}[ht]
      \centering
      \includegraphics[width=0.9\textwidth]{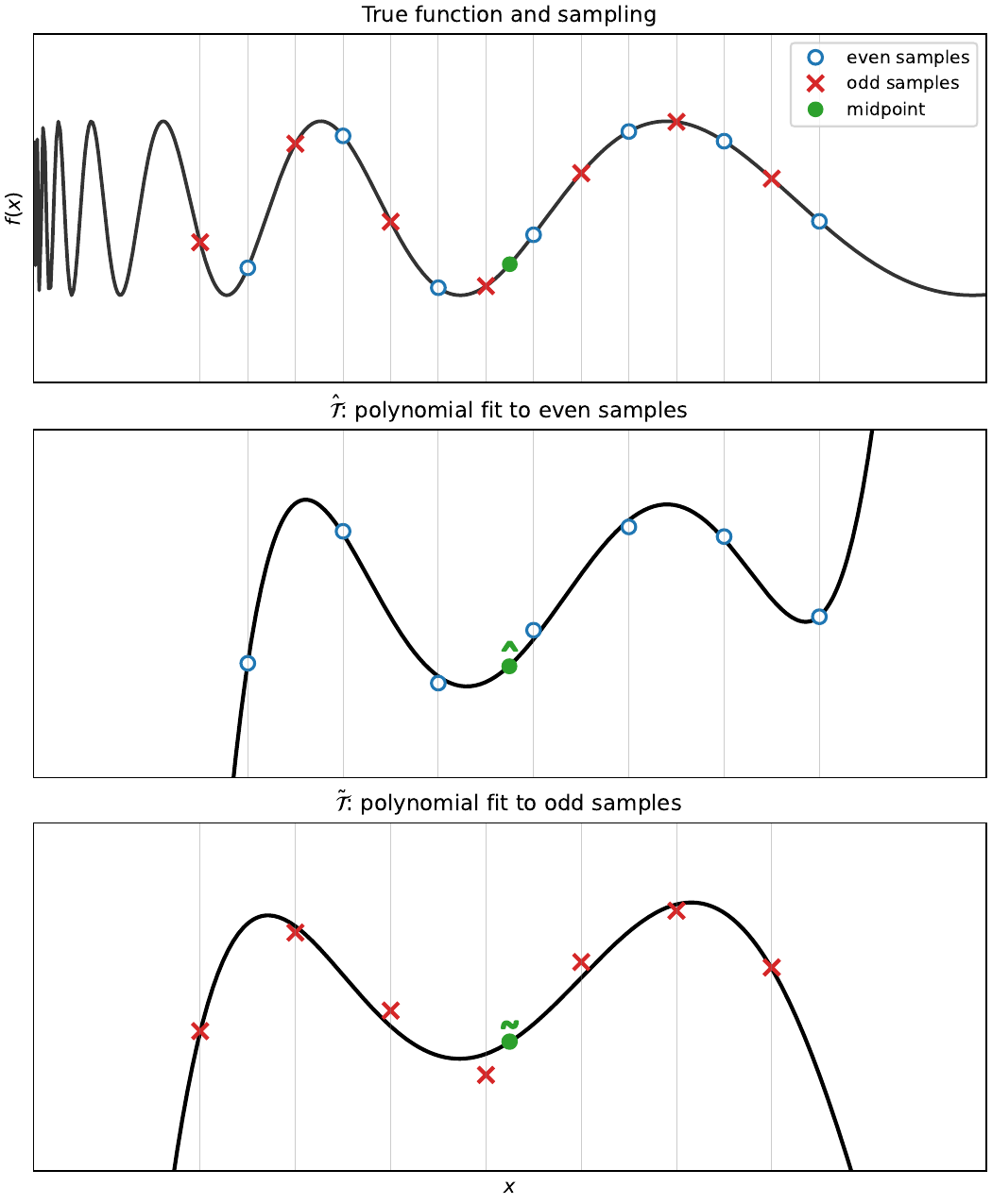}
      \caption{The top panel shows the true function with even (blue circles) and odd (red crosses) samples. The middle panel shows \(\hat{\mathcal T}\) constructed from even samples, and the bottom panel shows \(\tilde{\mathcal T}\) constructed from odd samples. Both filters interpolate to the same point ``\(\bullet\)'' using disjoint data subsets, resulting in estimates \(\hat{\bullet}\) and \(\tilde{\bullet}\) respectively. Note that from the points of views of \(\hat {\mathcal T}\) and \(\tilde {\mathcal T}\), the evaluation point is one-quarter step off center.}
      \label{fig:sample-split-filters}
\end{figure}

\FloatBarrier
\input{iv-sections-l4dc/iv-proofs}

\input{iv-sections-l4dc/filtering}
\input{iv-sections-l4dc/cor}
\input{iv-sections-l4dc/lorenz-supplement}
\input{iv-sections-l4dc/vanderpol-supplement}

\end{document}

%% file: introduction-l4dc.tex
\section{Introduction}
Filtering, smoothing, prediction, and control benefit from accurate models of a plant's dynamics.
Our present data-rich world facilitates and demands dynamics models that can learn efficiently from long time series.
Just as in large-scale machine learning, recurrent neural network architectures are being succeeded by autoregressive architectures such as transformers \citep{vaswani_attention_2023},
it is now preferable to specify an autoregressive model that maps past outputs to future outputs, rather than a state-space model that advances an unobserved state.
We further make the simplification \citep{brunton_discovering_2016,mezic_koopman_2021} of assuming a parametric form in which the prediction is nonlinear in the inputs (past outputs, exogenous inputs, etc.)
but linear in the parameter vector.
Thus the parameter vector may be estimated (following a feature selection process) by processing the data with a linear filter, applying a nonlinearity, and then running least squares estimation.
This is common practice in engineering \citep{kutz_dynamic_2016,haller_modeling_2025}.

But least-squares estimation is problematic when viewed as point estimation of the parameter vector.
Assuming that the data is contaminated by realistic levels of noise (such as from quantization \citep{maity_effect_2025}), 
least-squares regressions of (noisy) future outputs on (noisy)
past outputs are biased \citep[Chapter 8]{kutz_dynamic_2016} \citep{yi_handbook_2021,soderstrom_errors--variables_2018}.
For this reason, an autoregressive linear system identification method introduced in the 1970s has since been supplanted by instrumental variables (IV) and other bias-avoiding methods \citep[Chapter 1]{garnier_identification_2008} \citep{gonzalez_continuous-time_2022}.

The current wave of autoregressive system identification methods recognizes that output measurement noise degrades parameter estimation with a bias that persists even at the low noise levels achievable by judicious data prefiltering \citep{brunton_discovering_2016, wentz_derivative-based_2023,hsin_symbolic_2024}.
Therefore, it is not enough to filter away the noise and then pretend it is not there.
Two bias mitigation methods are bias correction, which depends strongly on accurate estimation of the noise variance, and instrumental variables, which does not \citep{garnier_identification_2008}.

We use instrumental variables to counteract the bias due to noise in the right-hand side of the model equation.
In IV problems in econometrics and engineering, an instrumental variable is provided along with the data or generated using model predictions \citep{davidson_econometric_2004,gonzalez_continuous-time_2022}, and one usually settles for asymptotic normality \citep{pan_efficiency_2020} because the IV estimator has no finite expectation \citep[\S8.4]{davidson_econometric_2004}.
\begin{after}
      \citet{zeiringer_instrumental_2026} use the model predictions to synthesize instruments for a linear-in-the-parameters nonlinear system, but it is hard to get statistical guarantees on the resulting estimator.
\end{after}
We show that under reasonable sampling assumptions, the instruments can be synthesized from the same data as the regressors.
\begin{after}
      Whereas the vanilla IV estimator's heavy tails 
      result from inverting a matrix whose probability mass may concentrate near singularity,
      our estimator imposes minor regularizations,
      leading to a finite expectation and consistency in \(L^p\) for all \(p \geq 1\).
\end{after}

\subsection*{Contributions}
We define an instrumental variables estimator by applying local polynomial regression \citep{de_brabanter_derivative_2013} in a new way 
at the data filtering step.
We analyze the benefits of singular value truncation, a form of ridge regularization, and instrumental variable truncation, which is a way to convert a subexponential tail to a subgaussian tail.
The final consistency result, stated in Theorem~\ref{thm:consistency}, requires delicate tail bounds originally developed for a biased estimator in \cite{kuang_estimation_2024-1}.

\subsection*{Outline}
We state the problem (discrete and continuous cases) in \S\ref{sec:problem}.
We construct the estimator in \S\ref{sec:algorithm} and state the theoretical principles behind its design.
We state the main theoretical result in Theorem \ref{thm:consistency}.
We apply it to the Lorenz system (discrete and continuous cases) in \S\ref{sec:lorenz}.

\section{Notation and conventions}
We write \(x \vee y\) for \(\max(x, y)\) and \(x \wedge y\) for \(\min(x, y)\).
\begin{after}
We write \(C^m(A, B)\) for the space of \(m\)-times continuously differentiable functions between Euclidean spaces \(A\) and \(B\).
\end{after}

Unless otherwise specified, the notation \(\left\|A\right\|\) refers
to the operator 2-norm of \(A\) or the \(2\)-norm of a vector.
A subscript denotes a stochastic \(L^q\) norm: \(\left\|x\right\|_{q}
= \del{\expect \left\|x\right\|^q}^{1/q}\).

The sub-gaussian norm of a random vector or matrix \(X\) is
\begin{gather*}
      \left\|X\right\|_{\psi_2}
       = \inf\cbr{t > 0: \expect
\exp\left(\frac{\left\|X\right\|^2}{t^2}\right) \leq 2}.
\end{gather*}
The centered sub-gaussian norm of a random vector or matrix \(X\) is
\(\left\|X -\expect X\right\|_{\psi_2}\).

The variable \(C\) denotes a constant that may depend on the
dimension of the problem (\(\mathsf{d_y}\) and \(\mathsf{d_\phi}\)),
and its value may change from line to line.
It never depends on the key statistical variables \(n\), \(N\), or \(h\).\footnote{Eliding
      dimensionality constants into \(C\) effectively commits our
      analysis to the ``classical'' low-dimensional regime, in which the design and parameter matrices are generic, full rank, etc.
      We reserve for future work the high-dimensional regime where
intrinsic dimension may be far less than the ambient dimension.}

%% file: iv-sections-l4dc/iv-analysis.tex
\section{Problem statement}
\label{sec:problem}
We adopt the Output Error model of a deterministic system with a single unknown \(\theta_0\) \citep[p.~31]{ljung_perspectives_2010}.
The data \(\{z_i\}_{i = 1}^{n}\) comprises measurements of the signal \(y\) at a sampling period \(h > 0\), plus noise \(\{\epsilon_i\}_{i = 1}^{n}\):
\begin{align}
      z_i = y(ih) + \epsilon_i, \quad i \in [1 \ldots n].
\end{align}
The time horizon is \(T = nh\).
The signal \(y \in C^m([0, T], \mathbb{R}^{\mathsf d_y})\) is taken as fixed (only the noise is random), and satisfies
\begin{align}
      \label{eq:model}
      \mathcal H y(t)
      &= \theta_0^\intercal \phi(t, \mathcal G y(t)),
\end{align}
where \(n\) is the number of observations,
\(h\) is the sampling period,
 \(m > 0\) is the smoothness of \(y\),
\(\phi \in C^1([0, T] \times \mathbb{R}^{\mathsf d _{\mathcal
G}}, \mathbb{R}^{\mathsf{d}_\phi})\)
is a static nonlinearity,
\(\theta_0 \in
\mathbb{R}^{\mathsf{d}_\phi \times \mathsf{d}_{\mathcal H}}\),
and
\(\mathcal H : C^m([0, T], \mathbb{R}^{\mathsf d_y})\to C([0, T], \mathbb{R}^{\mathsf d_{\mathcal H}})\)
and 
\(\mathcal G: C^m([0, T], \mathbb{R}^{\mathsf d_y})\to C([0, T], \mathbb{R}^{\mathsf d_{\mathcal G}})\)
 are \emph{known} linear operators with
compatible dimensions.
Our paper concerns the cases: (a) \emph{discrete-time}, where \(\mathcal H\) is the left shift \(\mathcal H x(t) = x(t + \tau)\) for some \(\tau > 0\) and \(\mathcal G\) is the identity; (b) first-order \emph{continuous-time}, where \(\mathcal H = \partial_t\) and \(\mathcal G\) is the identity; and (c) continuous-time autoregression, where \(\mathcal H\) is a high-order time derivative and \(\mathcal G\) contains lower-order derivatives.
\begin{after}
\begin{example}
      For example, a reduced-order model of vortex shedding behind a cylinder \citep[Eq.~6.1]{haller_modeling_2025},
      \begin{align*}
            \dot \rho
            &= 0.0584 \rho - 0.479 \rho^3 + 1.27 \rho^5 + 6.80 \rho^7 - 58.9 \rho^9 + 108 \rho^{11}
            \\
            \dot
            \gamma
            &= 0.553 + 0.441 \rho^2 - 3.38 \rho^4 + 55.5 \rho^6 - 321 \rho^8 + 626 \rho^{10}
      \end{align*}
      can be expressed using
      \(\mathcal H = \partial_t\),
      \(\mathcal G = \mathrm{id}\),
      \(y = (\rho, \gamma)\),
      and \(\phi(y) = (1, \rho, \rho^2, \ldots, \rho^{11})\).
\end{example}
\end{after}


\begin{problem}
      \label{prob:main}
      Given data
      \(\{z_i\}_{i \in [1 \ldots n]}\),
      find a plug-in estimator \(\hat\theta\) satisfying
      \begin{align*}
            \left\|\hat \theta - \theta_0\right\|_{q}
            &\leq \mathsf{error}(n, h, q)
      \end{align*}
      where \(\mathsf{error}(n, h, q)\) is a computable function satisfying \(\limsup_{h\to 0}\limsup_{n \to \infty} \mathsf{error}(n, h, q) = 0\) for all \(q\).
\end{problem}


\begin{assum}
      \label{assum:phi-lip}
      The nonlinearity \(\phi\) is Lipschitz in its second argument, uniformly in its first:
      \begin{gather*}
            \sup_{t \in [0, T]}
            \sup_{x \in \mathbb R^{\mathsf d_{\mathcal G}}}
            \left\|
                  \nabla_x \phi(t, x)
            \right\|
            < \infty.
      \end{gather*}
\end{assum}
\begin{assum}
      \label{assum:signal-regularity}
      For some \(p \geq 2\), the signal \(y\) satisfies
      \begin{gather*}
            \sup_{t \in [0, T]}
            \left\|
                  \dpd[p]{y}{t}(t)
            \right\|
            < \infty.
      \end{gather*}
\end{assum}
\begin{assum}
      \label{assum:noise-subg}
      The noise \(\{\epsilon_i\}_{i \in [1 \ldots n]}\) has zero mean, is independent across \(i\), and is subgaussian:
      \begin{gather*}
            \sup_{i \in [1 \ldots n]}
            \left\|
                  \epsilon_i
            \right\|_{\psi_2}
            \leq K < \infty.
      \end{gather*}
      
\end{assum}

\section{Design of the estimator}
\label{sec:algorithm}
We use two different kinds of regularization:
\begin{definition}[Regularization operators]
      \label{def:singular-value-clipping}
      Express \(A \in \mathbb R^{n \times n}\) as \(A = \sum_{i=1}^n \sigma_i u_i v_i^\intercal\) with \(\sigma_i \ge 0\) and \(\{u_i\}\), \(\{v_i\}\) orthonormal. The singular-value clipping operator is \(\clip[\lambda]{A} = \sum_{i=1}^n \max(\lambda, \sigma_i) u_i v_i^\intercal\). For \(\mu > 0\), \(\rho_\mu: \mathbb{R}^n \to \mathbb R^n\) is \(\rho_\mu(x) = x/(1 + \left\|x\right\|/\mu)\).
\end{definition}






The estimator has two hyperparameters \(\lambda, \mu > 0\).
We select regression times \(\{t_j\}_{j=1}^{n'}\) at which to impose \eqref{eq:model}, 
\begin{after}
      then approximate the \emph{continuous functions} \(\mathcal Hy: \mathbb R \to \mathbb R^{\mathsf d_{\mathcal H}}\) and \(\mathcal Gy: \mathbb R \to \mathbb R^{\mathsf d_{\mathcal G}}\) at these times using the filtered estimates \(\hat{\mathcal{H}}\), \(\hat{\mathcal{G}}\), and \(\tilde{\mathcal{G}}\) (described in Section~\ref{subsec:filters}).
\end{after}
Define the matrices \(Y \in\mathbb R^{n' \times \mathsf d_{\mathcal H}}\), \(X \in\mathbb R^{n' \times \mathsf d_{\phi}}\), and \(Z \in\mathbb R^{n' \times \mathsf d_{\phi}}\) by
\begin{align}
      Y_j &= \hat {\mathcal H} y(t_j),
      &
      X_j &= \phi(t_j, \hat {\mathcal G} y(t_j)),
      &
      Z_j &= \rho_\mu \circ \phi(t_j, \tilde {\mathcal G} y(t_j)),
\end{align}
where \(Y_j\), \(X_j\), and \(Z_j\) denote the \(j\)-th rows of the respective matrices.
Note that \(\tilde {\mathcal G} y(t_j)\) is stochastically independent of both \(\hat {\mathcal G} y(t_j)\) and \(\hat {\mathcal H} y(t_j)\), and each row of \(Z\) is bounded by \(\mu\) in absolute value.
The estimator is given by
\begin{align}
      \hat \theta
      &= \del{\clip {Z^\intercal X}}^{-1} Z^\intercal Y,
      %
\end{align}
and satisfies: 
\begin{corollary}
      \label{cor:consistency}
      Let \(X\), \(Y\), and \(Z\) come from the estimator \(\hat \theta\) defined in \S\ref{sec:algorithm}.
      Then \(\hat \theta\) satisfies the first two hypotheses of Theorem~\ref{thm:consistency}.
      If \(N\) has the ideal scaling with respect to \(h\), \(\sigma^2 \sim n\), \(\sup_i |y_i^\star| = \Theta(1)\), \(\sup_i |x_i^\star| = O(1)\),
      and 
      for all \(C\),
      \begin{align*}
            \exp\del{-C n^3 h^{2p/(2p + 1)}}
            \ll \lambda
            \ll n,
      \end{align*}
      then for any \(q \geq 1\),
      \begin{align*}
            \left\|\hat\theta - \theta_0\right\|_q
            &\lesssim
            h^{(p-d)/(2p + 1)} + \sqrt{
                  \frac{1}{
                        n h^{2p/(2p + 1)}
                  }
            }
      \end{align*}
      where \(d=1\) for the continuous-time estimator and \(d=0\) for the discrete-time estimator.
\end{corollary}
\begin{remark}[Sensitivity to tuning parameters]
The range of favorable \(\lambda = \lambda(n, h)\) as \(n\to \infty\), \(h \to 0\) is very wide; and at the \(\lambda \sim n\) extreme, the IV estimator behaves like least squares.
The role of \(\mu\) is more subtle.
On one hand, \(\mu\) multiplies the entire RHS of the error bound (hidden behind ``\(\lesssim\)''), so it would seem that smaller is better.
On the other hand, \(\mu\) also factors into the persistence of excitation condition \eqref{assumption:persistence}, and an excessively small \(\mu\) might make the condition fail.
Fortunately (in practice), the persistence of excitation condition is checkable in that \( \expect Z^\intercal X \approx Z^\intercal X \), which can be computed from the data.
\end{remark}

The proof of this result relies on Theorem~\ref{thm:consistency}, which has the form
\begin{align*}
      \text{\sffamily error} \lesssim 
      \del{\text{\sffamily sensitivity to }\lambda} \lambda
      + \del{\text{\sffamily sensitivity to data}} \del{\text{\sffamily bias} + \text{\sffamily noise}}.
\end{align*}
In order to motivate and satisfy the hypotheses of Theorem~\ref{thm:consistency}, we now narrate the differences between our estimator and the least-squares estimator \(\hat \theta_{\text{LS}} = (X^\intercal X)^{-1} X^\intercal Y\), which is obtained by replacing \(Z\) with \(X\), setting \(\lambda = 0\), and setting \(\mu = \infty\).

\paragraph{Why \(Z\)?}
\label{subsec:filters-why-Z}
As the upcoming example shows, the noise-noise interaction in \(X^\intercal X\) (and, to a lesser extent, in \(X^\intercal Y\)) is a source of bias.
Replacing \(X^\intercal X\) with \(Z^\intercal X\) and \(X^\intercal Y\) with \(Z^\intercal Y\) eliminates this source of bias; the columns of \(Z\) are called instrumental variables \citep[Chapter 8]{davidson_econometric_2004}.
We depict a one-dimensional miniature of the
quadratic terms in \(\hat \theta\) and \(\hat \theta_{\text{LS}}\):
\(\eta_1\) illustrates what is happening inside \(\hat\theta_\text{LS}\),
and \(\eta_2 \) illustrates what is happening inside \(\hat\theta\).
\begin{example}
      Suppose we are trying to estimate \(\eta = \frac{1}{n}\sum_{i =
      1}^n \mu_i^2 \in \mathbb{R}\) from the datasets \(\{X_{\mathsf
      A, i}\}_{i=1}^n\) and \(\{X_{\mathsf B, i}\}_{i=1}^n\)
      where for \(\beta \in \{\mathsf A, \mathsf B\}\) we have
      \(X_{\beta, i} = \mu_i + \epsilon_{\beta, i}\).
      The noise \(\epsilon_{\beta, i}\) has distribution
      \(\mathcal{N}(0, \sigma^2)\), and \(\epsilon_{\beta, i}\) is
      independent of \(\epsilon_{\beta', i'}\) if \(\beta \neq
      \beta'\) or \(i \neq i'\).
      Consider the two estimators
      \begin{align*}
            \hat \eta_1 &= \frac{1}{2n} \del{\sum_{i = 1}^{n}
            X_{\mathsf A, i}^2 + \sum_{i = 1}^{n} X_{\mathsf B, i}^2}
            &\text{and}&
            &
            \hat \eta_2 &= \frac{1}{n} \sum_{i = 1}^{n} X_{\mathsf A,
            i} X_{\mathsf B, i}
            \intertext{with means}
            \expect \hat \eta_1 &= \eta + \sigma^2
            &\text{and}&
            &
            \expect \hat \eta_2 &= \eta
            \intertext{and variances}
            \var \hat \eta_1 &= \var \hat\eta_2 = \frac{2\sigma^2
            \eta + \sigma^4}{n}.
      \end{align*}
      While these two estimators have identical variances, \(\hat
      \eta_2\) is unbiased and has a strictly smaller MSE than \(\hat \eta_1\).
\end{example}

This example shows that the independence structure of \(Z^\intercal X\) is key.
We achieve it in our problem via a novel sample-split design for time series with a latent continuous-time structure.

\paragraph{The filters \(\hat {\mathcal H}\), \(\hat {\mathcal G}\), and \(\tilde {\mathcal G}\).}
\label{subsec:filters}

To approximate \(\mathcal H\) and \(\mathcal G\) from noisy discrete data, we use local polynomial regression \citep{fan_local_2003,de_brabanter_derivative_2013}.
Each filter is specified by a \emph{differentiation stencil} \(\mathbf D^{k, i_0, N}_{d, h} \in \mathbb{R}^N\), which is a vector of coefficients that linearly combines \(N\) consecutive measurements to approximate the \(d\)-th derivative at location \(i_0 h\).

\begin{definition}[Differentiation stencil]
      \label{def:differentiation-stencil}
      Fix a window size \(N \in \mathbb{N}\), step size \(h > 0\), derivative order \(d \in \mathbb{N}\), and location \(i_0 \in \mathbb{R}\).
      The differentiation stencil \(\mathbf D^{k, i_0, N}_{d, h} \in \mathbb{R}^N\) satisfies
      \begin{align}
            \label{eq:stencil-exactness}
            \sum_{k=1}^N \mathbf D^{k, i_0, N}_{d, h} f(kh) &= \dod[d]{f}{x}(i_0 h)
            \quad \text{for all polynomials } f \text{ of degree at most } p-1,
      \end{align}
      where \(p \geq d+1\) is a design parameter.
      Among all stencils satisfying \eqref{eq:stencil-exactness}, we select the one with minimum Frobenius norm (see \appendixref{section:filtering-details} for construction).
\end{definition}


\begin{lemma}[Local polynomial filtering]
      \label{lem:local-polynomial-filter}
      Let \(f \in C^p([0, T], \mathbb{R})\) be a univariate function with \(R_p := \sup_{t \in [0, T]} |f^{(p)}(t)| < \infty\), and let \(\{w_k\}_{k=1}^N\) be independent mean-zero noise with \(\max_k \|w_k\|_{\psi_2} \leq \nu\).
      Define the filtered estimate
      \begin{align*}
            \widehat{f^{(d)}}(i_0 h) := \sum_{k=1}^N \mathbf D^{k, i_0, N}_{d, h} \del{f(kh) + w_k}.
      \end{align*}
      Then the bias and fluctuation satisfy
      \begin{align*}
            \left| \expect \widehat{f^{(d)}}(i_0 h) - f^{(d)}(i_0 h) \right| &\leq C(p, i_0) R_p (Nh)^{p-d},
            \\
            \left\| \widehat{f^{(d)}}(i_0 h) - \expect \widehat{f^{(d)}}(i_0 h) \right\|_{\psi_2} &\leq C(p, i_0) \nu N^{-d - \frac{1}{2}} h^{-d}.
      \end{align*}
\end{lemma}
\begin{proof}
      This relies on Assumptions~\ref{assum:signal-regularity} and \ref{assum:noise-subg}.
      The construction and analysis are detailed in \appendixref{section:filtering-details}.
      The stencil coefficients are obtained by solving a minimum-norm problem subject to the exactness constraint \eqref{eq:stencil-exactness} (Proposition~\ref{prop:local-polynomial-filter}).
      The bias bound follows from Taylor expansion (Proposition~\ref{prop:local-polynomial-filter-bias}).
      The fluctuation bound uses subgaussian concentration (Proposition~\ref{prop:local-polynomial-filter-fluctuation}).
\end{proof}

We first ``unzip'' the time series into even and odd subsequences, then apply the local polynomial filter to each subsequence separately:

\begin{definition}[Sample-split filters]
      \label{def:sample-split-filters}
      Let \(\{z_i\}_{i=1}^n\) denote the measurements, and let \(z^{\text{even}} = \{z_{2j}\}_{j=1}^{\lfloor n/2 \rfloor}\) and \(z^{\text{odd}} = \{z_{2j-1}\}_{j=1}^{\lceil n/2 \rceil}\) denote the even- and odd-indexed subsequences.
      For a given operator \(\mathcal T\) with parameters \((i_0, d)\), define:
      \begin{align*}
            \hat{\mathcal T} y(t_j) &:= \sum_{k=1}^N \mathbf D^{k, i_0 - 1/4, N}_{d, 2h} z^{\text{even}}_k,
            &
            \tilde{\mathcal T} y(t_j) &:= \sum_{k=1}^N \mathbf D^{k, i_0 + 1/4, N}_{d, 2h} z^{\text{odd}}_k,
      \end{align*}
      where the stencils are evaluated at step size \(2h\) (twice the original sampling period) and shifted locations \(i_0 \pm 1/4\).
\end{definition}

Both \(\hat{\mathcal T}\) and \(\tilde{\mathcal T}\) approximate the same continuous function at the same off-grid location (corresponding to ``\(z_{i+0.5}\)''), but use disjoint subsets of the data; see Figure~\ref{fig:sample-split-filters} in \appendixref{section:sample-split-illustration} for a visualization that explains the quarter-step offset.

Because the noise terms \(\{\epsilon_i\}\) are independent across indices, we have:
\begin{lemma}
      \label{lem:filters-independence}
      The filters \(\hat{\mathcal H}\), \(\hat{\mathcal G}\), and \(\tilde{\mathcal G}\) obey the same bounds as in Lemma~\ref{lem:local-polynomial-filter}; moreover, \(\hat{\mathcal H}\) and \(\hat{\mathcal G}\) are independent of \(\tilde{\mathcal G}\).
\end{lemma}

\paragraph{Why \(\lambda\)?}
The vanilla IV estimator \((Z^\intercal X)^{-1} Z^\intercal Y\) is well-studied for its unbiasedness and asymptotic normality and efficiency under standard identifying assumptions \citep{gonzalez_consistency_2024}.
However, it has a heavy-tailed distribution and does not have any finite moments of any order \citep[\S8.4]{davidson_econometric_2004}.
We seek to capture a finite-sample analog of the traditional asymptotic consistency result.
By \begin{after}
      clipping the lower spectrum of \(Z^\intercal X\),
\end{after} we ensure that for all \(q\geq 1\), \(\left\|\hat \theta\right\|_q < \infty\).
It is furthermore helpful that singular value clipping is a bounded perturbation of the unclipped matrix:
\begin{lemma}
      \label{lemma:singular value perturbation}
      Let \(A\) be any matrix. Then
      \(\left\|A - \clip A\right\| \leq \lambda\).
\end{lemma}
\begin{proof}
      Because \(A\) and \(\clip A\) have the same singular vectors, \(\left\|A - \clip A\right\|\) has singular values at most \(\lambda\).
\end{proof}

\paragraph{Why \(\mu\)?}
\label{subsec:filters-why-mu}
We have made \(\left\|\hat \theta\right\|_q\) finite by choosing \(\lambda > 0\).
We must now go further and ensure that \(\left\|\hat \theta - \theta_0\right\|_q\) is small;
in order to prove the IV estimator's consistency, we need a concentration inequality that bounds the deviations of \(Z^\intercal X\) from a nominal value.
If we set \(\mu =\infty\), then both \(Z\) and \(X\) would have sub-gaussian rows, and one could apply a Hanson-Wright inequality to get sub-exponential concentration of \(Z^\intercal X\) \citep{ziemann_tutorial_2023}.
However, it is easier to work with sub-gaussian concentration than sub-exponential concentration, especially when working in probabilistic \(L^p\) spaces.
Hence, by choosing \(\mu < \infty\), the rows of \(Z\) become bounded by construction, and we can get sub-gaussian concentration of \(Z^\intercal X\), and the following moment inequality for its regularized inverse:

\begin{proposition}[Three-way integral]
      \label{lem:three-way-integral}
      Let \(0 < a < b < \infty\) be constants.
      Let \(W>0\) be a real-valued random variable satisfying
      \(\probability\del{W \geq t} \leq \exp(-t^2/K^2)\) for some \(K > 0\).
      Then for all \(r \geq 1\), \(X\) defined by
      \begin{align*}
            X &= \frac{1}{a + 0 \vee (b - W)}
            \\
            \intertext{satisfies}
            \left\|X\right\|_r
            &\leq
            \gamma(r; a, b, K) = \gamma_{\text{head}}(r; a, b) + \gamma_{\text{body}}(r; a, b, K) + \gamma_{\text{tail}}(r; a, b, K),
            \intertext{where}
            \gamma_{\text{head}}(r; a, b)
            &= \frac{2}{b}
            \\
            \gamma_{\text{body}}(r; a, b, K)
            &=
            C\frac{r^{1/r}K^{1/r}}{b^{2(1+1/r)}}
            \exp\del{CK^2
                  \frac{
                        (r+ 1)^2\log^2 (a/b)
                  }{
                  r\del{b- a}^2}
            },
            \\
            \intertext{and}
            \gamma_{\text{tail}}(r; a, b, K)
            &=
            \frac{\exp(-C b^2/r K^2)}{a}
      \end{align*}
      \begin{after}
            for some absolute constant \(C > 0\).
      \end{after}
\end{proposition}
\begin{proof}
      The full proof is given in \appendixref{sec:three-way-integral-proof}.
      The idea is to truncate \(X\) into three regions based on the support of \(W\): a ``head'' where \(W\) is small, a ``tail'' where \(W\) is very large, and (our innovation) a ``body'' where \(W\) is large but not too large.
      In this region, the value of \(X\) transitions from \(1/a\) to \(1/b\).
      We integrate over this region using a ``layer cake'' Tonelli rearrangement to get a Gaussian integral.
\end{proof}



%% file: iv-sections-l4dc/iv-theorem.tex
Finally, we state a technical result which bounds the \(L^p\) risk of \(\hat\theta = \clip[\lambda]{Z^\intercal X}^{-1} Z^\intercal Y\) under abstract assumptions on \(Z\), \(X\), and \(Y\).
\begin{theorem}
      \label{thm:consistency}
      Let \(\lambda > 0\), \(q \geq 2\), and \(\epsilon> 0\).
      Suppose that unobserved \(Y^\star\in \mathbb{R}^{n \times \mathsf d_y}\), \(X^\star\in \mathbb{R}^{n \times \mathsf d_x}\), and \(\theta^\star \in
      \mathbb{R}^{\mathsf d_x \times \mathsf d_y}\)
      satisfy
      \begin{align}
            \label{eq:regression-model}
            Y^\star &= X^\star \theta^\star.
      \end{align}
      Suppose we have data \((z_i, x_i, y_i)\) for \(i \in [1\ldots n]\), such that
      \begin{enumerate}
            \renewcommand{\theenumi}{\alph{enumi}}
            \item 
            \label{assum:data-bounds}
            The data bounds hold:
            \begin{align*}
                  \sup_{i \in [1\ldots n]}
                   \left\|
                        \expect z_i y_i^\intercal - z_i (y_i^\star)^\intercal
                  \right\|
                  &=
                  \bar \nu_{zy} < \infty ,
                  &
                   \sup_{i \in [1\ldots n]}
                  \left\|
                        \expect z_i y_i^\intercal - z_i y_i^\intercal
                  \right\|_{\psi_2}
                  &=
                  \tilde \nu_{zy} < \infty,
                  \\
                  \sup_{i \in [1\ldots n]}
                  \left\|
                        \expect z_i x_i^\intercal - z_i (x_i^\star)^\intercal
                  \right\|
                  &= \bar \nu_{zx} < \infty,
                  &
                  \sup_{i \in [1\ldots n]}
                  \left\|
                        \expect z_i x_i^\intercal - z_i x_i^\intercal
                  \right\|_{\psi_2}
                  &= \tilde \nu_{zx} < \infty.
            \end{align*}
            \item
            For \(i, i'\in [1\ldots n]\), if \(|i - i'| \geq N\),
            then \((x_i, y_i, z_i)\) is independent of \((x_{i'}, y_{i'}, z_{i'})\).
            \item
            The data satisfies the persistence of excitation condition:
                  \label{assumption:persistence}
                  \begin{align*}
                        \sigma_\text{min} \del{
                              \expect {Z}^\intercal X
                        }
                        \geq \sigma^2 > 0.
                  \end{align*}
      \end{enumerate}

      Then the estimator defined by
      \begin{align*}
            \hat \theta &= \del{
                  \clip[\lambda]{Z^\intercal X}
            }^{-1} \del{
                  Z^\intercal Y
            }
      \end{align*}
      satisfies
      \begin{multline*}
            \frac{\left\|\hat\theta - \theta^*\right\|_q}{\left\|\theta^\star\right\|}
            \leq 
            \gamma(q; \lambda, \sigma^2 - \lambda, \sqrt{nN} \tilde \nu_{zx}) \lambda
            +
            C \sqrt{q(1+1/\epsilon)} \gamma(q(1+1/\epsilon); \lambda, \sigma^2 - \lambda, \sqrt{nN} \tilde \nu_{zx})
            \\\cdot
            \Bigg\{
            n \del{
                  \bar \nu_{zx}
                  + \left\|\theta^\star\right\|^{-1} \bar \nu_{zy}
            }
            + \sqrt{nN} \del{
                  \tilde \nu_{zx}
                  + \left\|\theta^\star\right\|^{-1}
                  \tilde \nu_{zy}
            }
            \Bigg\}.
      \end{multline*}
      where \(\gamma\) is the function appearing in Proposition~\ref{lem:three-way-integral}.

\end{theorem}


%% file: iv-sections-l4dc/lorenz.tex
\section{Numerical examples}
\label{sec:lorenz}
Our data comes from the Lorenz system
with sinusoidal forcing:
\begin{align*}
      \dot x^1 &= \sigma(x^2 - x^1),
      &
      \dot x^2 &= x^1(\rho- x^3) - x^2,
      &
      \dot x^3 &= u(t) x^1x^2- \beta x^3,
      &
      u(t) &= \sin(2\pi f t)
\end{align*}
which we write as
\begin{align}
      \dot \xi &= \theta_0^\intercal \phi(t, \xi)
      \label{eq:lorenz}
\end{align}
where
\begin{align*}
      \xi &= (x^1, x^2, x^3)
      &
      \text{and}
      &&
      \phi(t, x^1, x^2, x^3)
      &=
      (\sin (2\pi f t), x^1, x^2, x^3, x^1x^2, x^1x^3).
\end{align*}
Since our contribution focuses on the regression step, rather than
(regularized or sequential) sparse model selection, we regard the entire matrix \(\theta_0 \in \mathbb{R}^{6 \times 3}\) as unknown.
See Appendix~\ref{section:lorenz-supplement} for details on the data, estimator, and reporting.

\subsection{Continuous-time}
\label{subsec:lorenz-continuous}
Taking \(\mathcal H = \pd{}{t}\) and \(\mathcal G = \operatorname{1}\) recovers the true data-generating process \eqref{eq:lorenz}.
We compare a sample-split IV estimator of our design to a least-squares estimator, taking \(\theta_0\) as ground truth.

Our estimator achieves a \(\sim\)200x reduction in bias and \(\sim\)2x reduction in \(L^2\) risk versus least squares (Table~\ref{fig:lorenz-continuous-error}); its sampling distribution appears nearly unbiased, while least squares is biased toward zero with less dispersion (Figure~\ref{fig:lorenz-continuous-distribution}, Appendix~\ref{section:lorenz-supplement}).

\begin{table}
      \centering
      \input{iv-koopman/output/lorenz-continuous-error.tex}
      \caption{\label{fig:lorenz-continuous-error}%
      Comparison of bias, standard deviation, and root mean square error of our estimator and a least squares estimator for the parameter of the continuous-time Lorenz system.
      Monte Carlo standard errors in parentheses.
      }
\end{table}

\subsection{Discrete-time}
Taking \(\mathcal H \) to be a left shift by \(h\) and \(\mathcal G = \operatorname{1}\) produces a first-order discretization of \eqref{eq:lorenz}.
We compare a sample-split IV estimator of our design to a least-squares estimator.
There is no ground truth in this case, so we compare against a pseudo-true value obtained by evaluating the least-squares estimator on noise-free data.

Our estimator achieves a \(\sim\)500x reduction in bias and \(\sim\)10x reduction in \(L^2\) risk versus least squares (Table~\ref{fig:lorenz-discrete-error}); the sampling distribution is again virtually unbiased, while least squares is biased toward zero (Figure~\ref{fig:lorenz-discrete-distribution}, Appendix~\ref{section:lorenz-supplement}).

\begin{table}
      \centering
      \input{iv-koopman/output/lorenz-discrete-error.tex}
      \caption{\label{fig:lorenz-discrete-error}%
      Comparison of bias, standard deviation, and root mean square error (normalized by the pseudo-true value) of our estimator and a least squares estimator for the parameter of the discrete-time Lorenz system.
      Monte Carlo standard errors in parentheses.
      }
\end{table}

%% file: iv-koopman/output/lorenz-continuous-error.tex
\begin{tabular}{lrrr}
\toprule
Estimator & abs. bias (\%) & std (\%) & rmse (\%) \\
\midrule
Instrumental Variables (ours) & \num{0.017 \pm 0.008} & \num{0.800 \pm 0.007} & \num{0.800 \pm 0.007} \\
Least Squares & \num{2.382 \pm 0.003} & \num{0.517 \pm 0.005} & \num{2.437 \pm 0.003} \\
\bottomrule
\end{tabular}

%% file: iv-koopman/output/lorenz-discrete-error.tex
\begin{tabular}{lrrr}
\toprule
Estimator & abs. bias (\%) & std (\%) & rmse (\%) \\
\midrule
Instrumental Variables (ours) & \num{0.00318 \pm 0.00148} & \num{0.14431 \pm 0.00126} & \num{0.14434 \pm 0.00126} \\
Least Squares & \num{1.51112 \pm 0.00059} & \num{0.07690 \pm 0.00073} & \num{1.51307 \pm 0.00059} \\
\bottomrule
\end{tabular}

%% file: iv-sections-l4dc/vanderpol.tex
\begin{after}
\subsection{Van der Pol oscillator}
\label{subsec:vanderpol}
Unlike the Lorenz examples, the Van der Pol oscillator exercises the case \(\mathcal G \neq \operatorname{1}\).
The scalar observation \(x(t)\) satisfies the second-order ODE
\begin{align*}
      \ddot x &= \mu(1 - x^2)\dot x - x,
\end{align*}
which we write as
\begin{align}
      \mathcal H\, x(t) &= \theta_0^\intercal \phi(\mathcal G\, x(t))
      \label{eq:vanderpol}
\end{align}
where \(\mathcal H = \partial_t^2\), \(\mathcal G = \del{\partial_t^0, \partial_t^1}\),
\begin{align*}
      \phi(x, \dot x) &= (x,\; \dot x,\; x^2 \dot x),
      &
      \theta_0 &= (-1,\; \mu,\; -\mu)^\intercal.
\end{align*}
Both \(\mathcal H\) and \(\mathcal G\) require derivative estimation from the noisy scalar data \(\{z_i\}\).
We compare a sample-split IV estimator (with \(\mu = 2\)) to a least-squares estimator, taking \(\theta_0\) as the ground truth.
See Appendix~\ref{section:vanderpol-supplement} for details on the data, estimator, and reporting.

Our estimator achieves a substantial reduction in bias and \(L^2\) risk versus least squares (Table~\ref{tab:vanderpol-continuous-error}, Appendix~\ref{section:vanderpol-supplement}); its sampling distribution appears nearly unbiased, while least squares is biased toward zero (Figure~\ref{fig:vanderpol-continuous-distribution}).
\end{after}

\begin{figure*}
      \includegraphics[width=1\linewidth]{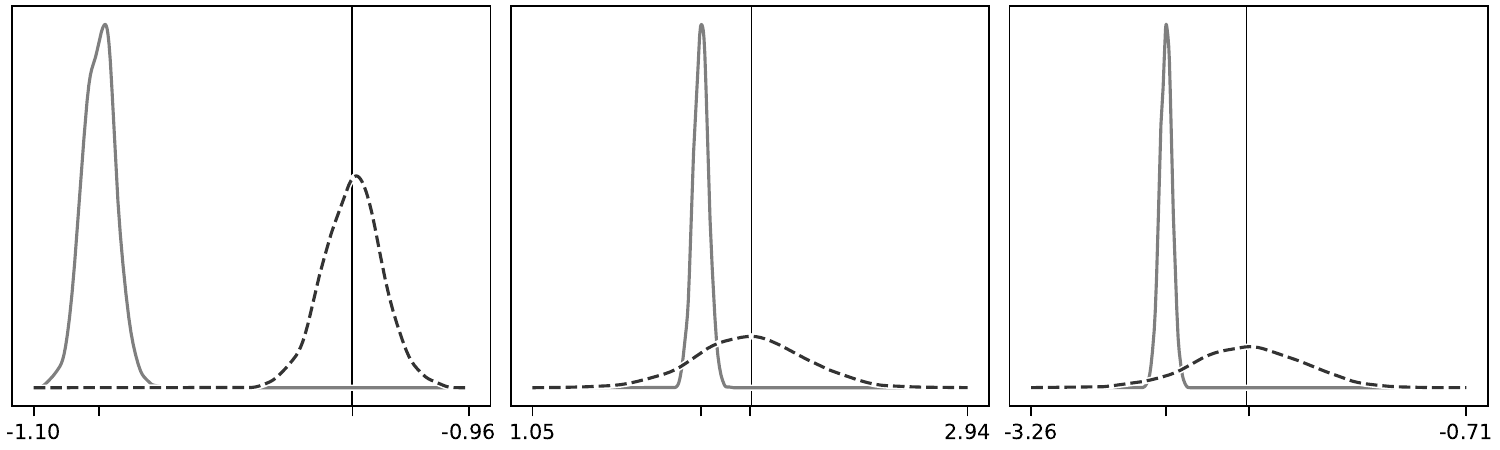}
      \caption{
            \label{fig:vanderpol-continuous-distribution}
            Elementwise marginal kernel density estimates of the sampling distributions
            of our estimator (dashed) and a baseline estimator (solid) for the Van der Pol parameter.
            Vertical line indicates ground truth; ticks indicate mean of sampling distribution.
      }
\end{figure*}

%% file: iv-sections-l4dc/iv-proofs.tex
\section{Proof of Proposition~\ref{lem:three-way-integral}}
\label{sec:three-way-integral-proof}
\begin{proof}
      Observe that \(a + 0 \vee (b - W)\) enjoys the following upper bound
      according to three cases of \(W\):
      \begin{align}
            a + 0 \vee (b - W)
            &\geq
            \begin{cases}
                  b, & \phantom{\lambda \leq{}} W \leq a \\
                  a + b - W, & a \leq W < b \\
                  a, & \phantom{\lambda \leq{}} b < W
            \end{cases}
      \end{align}
      These three cases yield an additive decomposition of \(X\) as
      \begin{gather}
            X \leq
            X_1 + X_2 + X_3,
            \quad \text{where}
      \end{gather}
      \begin{align}
            X_1 &= \frac{\bm{1}_{W \geq a}}{b},
            &
            X_2 &= \frac{\bm{1}_{a \leq W < b}}{a + b - W},
            &\text{and}
            &&
            X_3 &= \frac{\bm{1}_{W \geq b}}{a}.
      \end{align}
      \begin{align}
            \intertext{By the triangle inequality of \(L^r\) norms, we have}
            \left\|X\right\|_r
            &\leq
            \left\|X_1\right\|_r
            + \left\|X_2\right\|_r
            + \left\|X_3\right\|_r
            \label{eq:X-1-2-3}
            \\
            \intertext{By taking expectations, we immediately obtain the bounds}
            \left\|X_1\right\|_r
            &\leq \frac{1}{b}
            \label{eq:upper-bound-x1}
            \quad\text{and}
            \\
            \left\|X_3\right\|_r
            &\leq
            \frac{\probability\del{W \geq b}^{r}}{a}
            \leq \frac{\exp(-b^2/rK^2)}{a}.
            \label{eq:upper-bound-x3}
      \end{align}
      To bound \(\expect \left\|X_2\right\|^r\), we use the
      Fundamental Theorem of Calculus:
      \begin{align}
            X_2^r &= X_{2, 1}^r + X_{2,2}^r,
            \intertext{where}
            X_{2, 1}^r &=  \left. \frac{1}{\del{a + b - s}^r}
            \right|_{s = a} = \frac{1}{b}
            \label{eq:X-21}
            \\
            \intertext{results in}
            \left\|X_{2, 1}\right\|_r &\leq \frac{1}{b}
            \\
            \intertext{and}
            X_{2,2}^r
            &=
            \bm{1}_{a \leq W < b} \left. \frac{1}{\del{a + b - s}^r}
            \right|_{s = a}^{s = W}
            \\
            &=
            \bm{1}_{a \leq W < b}
            \int_{a}^{W}
            \dod{}{s}\sbr{
                  \frac{1}{
                        \del{a + b - s}^r
                  }
            }
            \dif s
            \\
            \intertext{Taking expectations,}
            \expect X_{2,2}^r
            &=
            \expect \bm{1}_{a \leq W < b}
            \int_{a}^{W}
            \frac{r}{
                  \del{a + b - s}^{r + 1}
            }
            \dif s
            \\
            &=
            \expect \int_{a}^{b}
            \frac{r \bm{1}_{W \geq s}}{
                  \del{a + b - s}^{r + 1}
            }
            \dif s
            \\
            &=
            \int_{a}^{b}
            \frac{r \probability\del{W \geq s}}{
                  \del{a + b - s}^{r + 1}
            }
            \dif s
            \tag{Tonelli}
            \\
            &=
            r \int_{a}^{b}
            \underbrace{
                  \frac{1}{
                        \del{a + b - s}^{r + 1}
                  }
            }_{=: g(s)}
            e^{-s^2/K^2}
            \dif s
            \tag{Lemma~\ref{lem:concentration}}
            \\
            &=
            r \int_{a}^{b}
            e^{-s^2/K^2 + \log g(s)}
            \dif s
            \\
                  &\leq
                  r \int_{a}^{b}
                  \exp\del{-\frac{s^2}{K^2}
                        + \log(b^{-2(r+1)})
                        +
                        \frac{\log\sbr{(a/b)^{-(r+1)}}}{b
                        - a} s}
                  \dif s
            \tag{by convexity of \(\log g(s)\) on \([a, b]\)}
            \\
            &\leq
            \frac{r}{b^{2(r+1)}}
            \int_{-\infty}^{\infty}
            \exp\del{
                  -\frac{s^2}{K^2}
                  +
                  \frac{
                        \log\sbr{(a/b)^{-(r+1)}}
                  }{
                  b- a}
                  s
            }
            \dif s
            \\
            &=
            C\frac{rK}{b^{2(r+1)}}
            \exp\del{CK^2
                  \frac{
                        (r+ 1)^2\log^2 (a/b)
                  }{
                  \del{b- a}^2}
            }
            \intertext{by the Gaussian integral identity
             \(
             \int_{-\infty}^{\infty}
             e^{- (a x^2 + b x)}
             \dif x
             =\sqrt{\frac{\pi}{a}} e^{\frac{b^2}{4a}}
             \).
            Raising both sides to the power \(1/r\),}
            \left\|X_{2,2}\right\|_r
            &\leq
            C\frac{r^{1/r}K^{1/r}}{b^{2(1+1/r)}}
            \exp\del{CK^2
                  \frac{
                        (r+ 1)^2\log^2 (a/b)
                  }{
                  r\del{b- a}^2}
            }
            \label{eq:X-22}
      \end{align}
      by a Gaussian integral.
      We conclude a bound on \(\|X_{2,2}\|_r\) by raising both sides
      to the power \(\frac{1}{r}\).

      We re-associate the summands in \eqref{eq:X-1-2-3},
      \begin{align}
            \left\|X\right\|_r
            &\leq \underbrace{
                  \left\|X_1\right\|_r + \left\|X_{2,1}\right\|_r
            }_{:= \gamma_\text{head} }
            + \underbrace{\left\|X_{2,2}\right\|_r}_{\gamma_\text{body}}
            + \underbrace{\left\|X_3\right\|_r}_{\gamma_\text{tail}}.
      \end{align}
      and conclude by inserting \eqref{eq:upper-bound-x1} for \(\left\|X_1\right\|_r\), \eqref{eq:X-21} for \(\left\|X_{2,1}\right\|_r\), \eqref{eq:X-22} for \(\left\|X_{2,2}\right\|_r\), and \eqref{eq:upper-bound-x3} for \(\left\|X_3\right\|_r\).

\end{proof}

\section{Proof of Theorem~\ref{thm:consistency}}
\label{sec:consistency-proof}
Manipulating \eqref{eq:regression-model}
yields the identity
\begin{align}
      Z^\intercal Y
      &= 
      Z^\intercal Y^\star
      + \sbr{
            Z^\intercal Y - Z^\intercal Y^\star
      }
      \\
      &= 
      (Z^\star)^\intercal X^\star \theta^*
      + \sbr{
            Z^\intercal Y - Z^\intercal Y^\star
      }
      \\
      \begin{split}
      &= 
      \clip[\lambda]{Z^\intercal X} \theta^*
      + \sbr{
            Z^\intercal X
            - \clip[\lambda]{Z^\intercal X}
      } \theta^\star
      \\
      &\quad
      +
      \sbr{
            (Z^\star)^\intercal X^\star
            - Z^\intercal X
      } \theta^\star
      + \sbr{
            Z^\intercal Y - Z^\intercal Y^\star
      }
      \end{split}
\end{align}
Multiplying both sides by the inverse of \(S = \clip[\lambda]{Z^\intercal X}\) and inserting the definition of \(\hat\theta\), we have the decomposition
\begin{gather}
      \hat\theta - \theta^*
      = 
      S^{-1}
      \Bigg\{\sbr{
            Z^\intercal X
            - \clip[\lambda]{Z^\intercal X}
      } \theta^\star
      +
      \sbr{
            Z^\intercal X^\star
            - Z^\intercal X
      } \theta^\star
      + \sbr{
            Z^\intercal Y - Z^\intercal Y^\star
      }
      \Bigg\},
\end{gather}
which may be combined with the Hölder conjugacy
\begin{align}
      \frac{1}{q} &= \frac{1}{q(1+ \epsilon^{-1})} + \frac{1}{q(1 + \epsilon)},
\end{align}
to yield
\begin{multline}
      \frac{\left\|\hat\theta - \theta^*\right\|_q}{\left\|\theta^\star\right\|}
      \leq 
      \underbrace{
            \left\|S^{-1}\right\|_{q}
      }_{\text{Lemma~\ref{lemma:combined-bound}}}
      \underbrace{
            \left\|
                  Z^\intercal X
                  - \clip[\lambda]{Z^\intercal X}
            \right\|_{\infty}
      }_{\text{Lemma~\ref{lemma:singular value perturbation}}}
      \\
      +
      \underbrace{
            \left\|S^{-1}\right\|_{q}
      }_{\text{Lemma~\ref{lemma:combined-bound}}}
      \Bigg\{
      \underbrace{
            \left\|
                  Z^\intercal X^\star
                  - Z^\intercal X
            \right\|_{q (1 + \epsilon)}
      }_{\text{Lemma~\ref{lem:concentration} and Fact~\ref{fact:subgaussian-tail}}}
      +
      \left\|\theta^\star\right\|^{-1}
      \underbrace{
            \left\|
                  Z^\intercal Y - Z^\intercal Y^\star
            \right\|_{q (1 + \epsilon)}
      }_{\text{Lemma~\ref{lem:concentration} and Fact \ref{fact:subgaussian-tail}}}      
      \Bigg\}
\end{multline}
To finish the proof,
we use Lemma~\ref{lemma:combined-bound} to bound the prefactor as a function of \(\sigma^2\), \(\lambda\), and \(q(1+1/\epsilon)\):
\begin{multline}
      \gamma(s; \sigma^2, \lambda) =
      \frac{2}{\sigma^{2} - \lambda}
      +
      \frac{s^{1/s} CL_{ZX}}{
            \del{\sigma^{2} - \lambda}^{2(1 + 1/s)}
      }
      \exp\del{CL_{ZX}^2
            \frac{
                  (s+ 1)^2\log^2 (\lambda/(\sigma^2-\lambda))
            }{s
            \del{\sigma^2- 2\lambda}^2}
      }
      \\
      +
      \frac{\exp(-C(\sigma^2 - \lambda)^2/s L_{ZX}^2)}{\lambda}.
\end{multline}
We invoke Lemma~\ref{lemma:singular value perturbation} to bound \(Z^\intercal X - \clip[\lambda]{Z^\intercal X}\) almost surely.
We invoke Lemma~\ref{lem:concentration} to bound the subgaussian norms of \((Z^\star)^\intercal X^\star
                  - Z^\intercal X\) and \(Z^\intercal Y - (Z^\star)^\intercal Y^\star\), and then use Fact~\ref{fact:subgaussian-tail} to convert these into \(L^p\) norms.
The result is
\begin{multline}
      \frac{\left\|\hat\theta - \theta^*\right\|_q}{\left\|\theta^\star\right\|}
      \leq 
      \gamma(q; \sigma^2, \lambda) \lambda
      +
      C \sqrt{q(1+1/\epsilon)} \gamma(q(1+1/\epsilon); \sigma^2, \lambda)
      \\\cdot
      \Bigg\{
      n \del{
            \bar \nu_{zx}
            + \left\|\theta^\star\right\|^{-1} \bar \nu_{zy}
      }
      + \sqrt{nN} \del{
            \tilde \nu_{zx}
            + \left\|\theta^\star\right\|^{-1}
            \tilde \nu_{zy}
      }
      \Bigg\}
\end{multline}

\subsection{Bounds on \texorpdfstring{\(N\)}{N}-dependent sums}
Before entering the proof of Theorem~\ref{thm:consistency}, we first recall some basic facts about (non-isotropic) vector- and matrix-valued subgaussian random variables; these are trivially adapted from the facts found in \cite[Chapter 2]{vershynin_high-dimensional_2018}.

\begin{fact}
      \label{fact:subgaussian-tail}
      Let \(X\) be a vector- or matrix-valued random variables with \(\left\|X\right\|_{\psi_2} = K\).
      Then
      \begin{enumerate}
            \renewcommand{\labelenumi}{\roman{enumi}.}
            \item
            The moments of \(X\) satisfy
            \begin{align*}
                  \left\|X\right\|_p
                  &\leq C K \sqrt{p}.
            \end{align*}

            \item The tails of \(X\) satisfy
                  \begin{gather*}
                  \probability \del{
                        \left\|X\right\| \geq t
                  }
                  \leq \exp \del{
                        -\frac{t^2}{CK^2}
                  }.
            \end{gather*}
      \end{enumerate}
\end{fact}

\begin{fact}[{\cite[Proposition 2.6.1]{vershynin_high-dimensional_2018}}]
      \label{fact:subgaussian-sum}
      Let \(\{X_i\}_{i = 1}^n\) be a sequence of  vector- or matrix-valued random variables.
      Then
      \begin{gather*}
            \left\|
            \sum_{i = 1}^n X_i
            \right\|_{\psi_2}^2
            \leq C \sum_{i= 1}^n \left\|X_i\right\|_{\psi_2}^2.
      \end{gather*}
\end{fact}

\begin{proposition}[Local dependence in \(L^q\)]
      \label{prop:local-dependence-q}
      Let \(q \geq 2\).
      Suppose that \(\{X_i\}_{i = 1}^n\) are random variables
      satisfying
      \begin{align*}
            \expect X_i &= 0, &&\forall i \in [n]
            \\
            \max_{i \in [n]} \left\|X_i\right\|_{\psi_2} &= \nu, &&\forall i \in [n]
            \\
            X_i &\perp X_j  &&\forall i, j \in [n] \text{ with } |i -j| \geq N
      \end{align*}
      Then
      \begin{align*}
            \left\|
            \sum_{i = 1}^n X_i
            \right\|_{\psi_2}
            &\leq C \sqrt{n N} \nu.
      \end{align*}
\end{proposition}
\begin{proof}
      For \(k \in [N]\), define the index sets
      \(I_k = \cbr{i \in [n]: i \cong k\ (\operatorname{mod} N)}\).
      \begin{align}
            \sum_{i = 1}^n X_i
            &= \sum_{k = 1}^N \sum_{i \in I_k} X_i
            \\
            \left\|
            \sum_{i = 1}^n X_i
            \right\|_{\psi_2}
            &\leq \sum_{k = 1}^N
            \left\|\sum_{i \in I_k} X_i\right\|_{\psi_2}
            \tag{triangle inequality}\\
            &\leq C \sum_{k = 1}^N
            \sqrt{|I_k|}\nu
            \tag{Fact~\ref{fact:subgaussian-sum}}
            \\
            &\leq C \sqrt{n N} \nu
      \end{align}
\end{proof}

\subsection{Term-by-term bounds}

\begin{lemma}
      \label{lem:concentration}
      The random matrices \(Z^\intercal X\) and \(Z^\intercal Y\) satisfy
      \begin{align*}
            \left\|
                  Z^\intercal X - \expect Z^\intercal X
            \right\|_{\psi_2}
            &\leq C \sqrt{nN} \tilde \nu_{zx}
            \\
            \left\|
            Z^\intercal X - Z^\intercal X^\star
            \right\|_{\psi_2}
            &\leq C\del{
                  n \bar \nu_{zx}
                  + \sqrt{nN} \tilde \nu_{zx}
            }
            \\
            \left\|
            Z^\intercal Y - Z^\intercal Y^\star
            \right\|_{\psi_2}
            &\leq C \del{
                  n \bar \nu_{zy}
                  + \sqrt{nN} \tilde \nu_{zy}
            }
      \end{align*}
\end{lemma}
\begin{proof}
      This follows from combining assumption \ref{assum:data-bounds}
      with Proposition~\ref{prop:local-dependence-q}.
\end{proof}

\begin{lemma}[Bounds on \(S^{-1}\)]
      \label{lemma:combined-bound}
      For all \(r \geq 1\), the matrix \(S^{-1}\) satisfies
      \begin{align*}
            \left\|S^{-1}\right\|_r
            &\leq
            \gamma(r; \lambda, \sigma^2, \sqrt{nN} \tilde \nu_{zx})
      \end{align*}
      where \(\gamma\) is the function defined in Proposition~\ref{lem:three-way-integral}.
\end{lemma}
\begin{proof}
      By the SVD,
      \begin{align}
            \left\|S^{-1}\right\| =
      \frac{1}{\sigma_\text{min}(S)},
      \label{eq:inverse-bound}
      \end{align}
      so our next task is to bound \(\sigma_\text{min}(S)\) from below.
      We have 
      \begin{align}
            \sigma_\text{min}(S) &=
            \max\del{
                  \lambda, \sigma_\text{min}(Z^\intercal X)
            }
            \tag{by \(S = \clip[\lambda]{Z^\intercal X}\)}\\
            &=
            \lambda +
            \del{
                  \sigma_\text{min}(Z^\intercal X) - \lambda
            }_+
            \tag{by the identity \(\max(a, b) = a + (b - a)_+\)}
            \\
            &\geq
            \lambda +
            \del{
                  \sigma_\text{min}\del{
                        \expect Z^\intercal X
                  } - \lambda - \sigma_\text{max}\del{
                        Z^\intercal X - \expect Z^\intercal X
                  }
            }_+
            \tag{Weyl's inequality}
            \\
            &\geq
            \lambda +
            \del{
                  \sigma^2 - \lambda - D
            }_+, \quad D = \sigma_\text{max}\del{
                        Z^\intercal X - \expect Z^\intercal X
                  }
            \tag{persistence of excitation hypothesis}
            \\
      \end{align}
      By Lem~\ref{lem:concentration},
      \(Z^\intercal X - \expect Z^\intercal X\) is subgaussian with constant \(C\sqrt{nN} \tilde \nu_{zx}\).
      By Fact~\ref{fact:subgaussian-tail}, for all \(t\geq 0\)
      \begin{align*}
            \probability\del{
                  \left\|Z^\intercal X - \expect Z^\intercal X\right\| \geq t
            }
            &\leq \exp\del{
                  -\frac{t^2}{\left(C\sqrt{nN} \tilde \nu_{zx}\right)^2}
            }.
      \end{align*}
      Now \eqref{eq:inverse-bound} becomes
      \begin{align}
            \left\|\hat S^{-1}\right\|
            &\leq \frac{1}{\lambda + (\sigma^2 - \lambda- D)_+},
      \end{align}
      which is amenable to Lemma~\ref{lem:three-way-integral} with constants
      \(a= \lambda\), \(b = \sigma^2\), and \(K = L_{ZX} = \sqrt{nN} \tilde \nu_{zx}\).
\end{proof}

%% file: iv-sections-l4dc/filtering.tex
\section{Construction of local polynomial filters}
\label{section:filtering-details}

\begin{proposition}
\label{prop:local-polynomial-filter}
For any window size \(N > 0\), step size \(h > 0\) and location \(i_0 \in \mathbb{R}\), there exist coefficients \(\mathbf D^{k, i_0, N}_{d, h}\), \(d \in [0\ldots m]\), \(k \in [1\ldots N]\), such that for all polynomials \(f\) of degree at most \(p - 1 < N\),
\begin{align*}
    \dod[d]{}{x}f(i_0h)
    &= \sum_{k = 1}^N \mathbf D^{k, i_0, N}_{d, h} f(kh).
    \intertext{Considered as a matrix in \((d, k)\),}
    \left\|
      \mathbf D^{\cdot, i_0, N}_{\cdot, h}
    \right\|
    &\leq C(p, i_0) N^{-m - \frac{1}{2}} h^{-m}.
\end{align*}
\end{proposition}
\begin{proof}
We prescribe \(\mathbf D\) as a solution to the following convex program:
\begin{align}
  \label{eq:D-program}
  \begin{split}
     & \min_{\mathbf D \in \mathbb{R}^{m \times N}} \quad \left\|\mathbf D\right\|_{\frob}
    \\
     & \operatorname{subject\ to}\quad
    \mathbf D A = B
  \end{split}
\end{align}
where \(A \in \mathbb{R}^{N \times p}\) and \(B \in \mathbb{R}^{(m+1) \times p}\) are given by
\begin{subequations}
  \label{eq:D-program-constraints}
  \begin{align}
    A_{ij}
     & =
    \left. \frac{(x - i_0h)^j}{N^j}\right|_{x = ih}
    = \frac{(i - i_0)^j h^j}{N^j}
    \\
    B^d_{j}
     & =
    \left. \frac{\dod[d]{}{x} (x - i_0h)^j}{N^j}\right|_{x = i_0h}
    = \delta_{dj} \frac{d!}{N^d h^d}
    \\
    i
     & \in [1\ldots N]
    \\
    j
     & \in [0\ldots p - 1]
    \\
    d
     & \in [0\ldots m]
  \end{align}
\end{subequations}

  \textbf{Explicit solution}
  Write the Frobenius inner product as \(\left\langle X, Y\right\rangle_{\frob} = \trace \del{X^\intercal Y}\).
  Let \(\Lambda \in \mathbb{R}^{(m + 1)\times p}\) be a Lagrange multiplier, and
  form the Lagrangian \(\frac{1}{2} \left\langle D, D\right\rangle_{\frob} - \left\langle \Lambda, DA - B\right\rangle_{\frob}\).
  First-order optimality yields \(D = \Lambda A^\intercal\).
  Right-multiplying by \(A\), we get \(B = \Lambda (A^\intercal A)\) which can be solved for \(\Lambda\).
  The result is the min-norm solution  \(D = B (A^\intercal A)^{-1} A^\intercal\).
  
  To bound \(D\), use
  \begin{align}
    \left\|D\right\|
     &\leq \left\| B \right\| \left\|(A^\intercal A)^{-1}A^\intercal\right\| 
     \\
     &\leq \frac{\left\| B \right\|}{\sigma_\text{min} (A)}
  \end{align}

  \textbf{Estimates}
  To estimate \(\sigma_\text{min} (A) = \lambda_\text{min} (A^\intercal A)^{1/2}\), notice that
  \begin{align}
    \del{A^\intercal A}_{jk}
     & = \sum_{i = 1}^N \del{\frac{i - i_0}{N}}^{j + k}
    \\
    \intertext{is a right Riemann sum. Evaluating the integral (with an error estimate),}
        \del{\tilde A^\intercal \tilde A}_{jk}
     & = \frac{N}{j + k + 1} + S_{jk}                   \\
    \left|S_{jk}\right|
     & \leq \frac{2p}{N}.
    \label{eq:hilbert-matrix}
    \intertext{As a consequence of this rescaling, we have the estimate}
    \left\| \del{\tilde A^\intercal \tilde A}^{-1} \right\|
     & \leq C(p) N^{-1}.
    \label{eq:tilde-A-bound}
  \end{align}
  \end{proof}

\begin{remark}[Numerics of \(\mathbf D\)]
  Note that
   \(A^\intercal A\) is a notoriously ill-conditioned Hilbert matrix \eqref{eq:hilbert-matrix}.
  For numerical stability, we solve for \(\mathbf D\) by rewriting the conditions \eqref{eq:D-program-constraints} in a basis of Legendre polynomials.
\end{remark}

\begin{proposition}[Bias]
\label{prop:local-polynomial-filter-bias}
  For some \(i_0\), \(h\), and \(N\), let \(\mathbf D\) be result of Proposition~\ref{prop:local-polynomial-filter}.
  Let \([x_0, x_1]\) be an interval and \(f \in C^p([x_0, x_1], \mathbb R)\) with \(R_p := \sup_{t \in [x_0, x_1]} |f^{(p)}(t)|\).
  Assume that \(d < p\).
  Then
  \begin{align*}
    \left|
      \dod[d]{f}{x}(x + i_0h)
      - \sum_{k = 1}^N \mathbf D^{k, i_0, N}_{d, h} f(x + kh)
    \right|
     & \leq
    C(m, p) R_p (Nh)^{p - m}.
  \end{align*}
\end{proposition}
\begin{proof}
  Expanding around \(x + i_0h\),
  \begin{align}
    f(x + kh)
     & = \sum_{\nu = 0}^{p - 1}
    \frac{f^{(\nu)}(x + i_0h)}{\nu!}((k - i_0)h)^\nu + R(k),
    \\
    |R(k)|
     & \leq
    C(m, p)\, R_p\, (Nh)^{p}.
    \label{eq:D-remainder}
  \end{align}
  Contracting the \(d\)th row of \(\mathbf D\) with \(\{f(x + kh)\}_{k=1}^N\),
  \begin{align}
    \sum_{k = 1}^N \mathbf D^k_d f(x + k h)
     & =
    \sum_{\nu = 0}^{p - 1}
    \frac{f^{(\nu)}(x + i_0h)}{\nu!}
    \sum_{k = 1}^N \mathbf D^k_d ((k - i_0)h)^\nu
    + \sum_{k = 1}^N \mathbf D^k_d R(k),
    \\
     & =
    f^{(d)}(x + i_0 h) + \sum_{k = 1}^N \mathbf D^k_d R(k),
  \end{align}
  where the last equality uses the constraints \eqref{eq:D-program-constraints}.
  Therefore,
  \begin{align}
    \left|
      f^{(d)}(x + i_0 h) - \sum_{k = 1}^N \mathbf D^k_d f(x + k h)
    \right|
     & \leq
    \left\|\mathbf D_{\cdot, d}\right\|_2 \, \left\| R(\cdot) \right\|_2
    \\
    \intertext{By Cauchy--Schwarz, $\big|\sum_{k=1}^N \mathbf D^k_d R(k)\big| = |\langle \mathbf D_{\cdot,d}, R\rangle| \le \|\mathbf D_{\cdot,d}\|_2 \, \|R\|_2$. Moreover, $\|R\|_2 \le \sqrt{N}\,\sup_k |R(k)|$ and $\|\mathbf D_{\cdot,d}\|_2 \le \|\mathbf D\|$, yielding:}
    & \leq
    \left\|\mathbf D\right\| \, \sqrt{N} \, \sup_k |R(k)|.
  \end{align}
  By Proposition~\ref{prop:local-polynomial-filter}, \(\left\|\mathbf D\right\| \leq C(p, i_0) N^{-m - \frac{1}{2}} h^{-m}\).
  Combining with \eqref{eq:D-remainder} yields the claim.
\end{proof}

\begin{proposition}[Fluctuation]
  \label{prop:local-polynomial-filter-fluctuation}
  For some \(i_0\), \(h\), and \(N\), let \(\mathbf D\) be the result of Proposition~\ref{prop:local-polynomial-filter}.
  Let \(\{w_k\}_{k=1}^N\) be independent, mean-zero, and subgaussian with \(\nu := \max_{k \in [1\ldots N]}\left\|w_k\right\|_{\psi_2} < \infty\).
  Then for any \(d < p\),
  \begin{align*}
    \left\|
      \sum_{k = 1}^N \mathbf D^{k, i_0, N}_{d, h} w_k
    \right\|_{\psi_2}
    &\leq C(p, i_0)\, \nu\, N^{-m - \frac{1}{2}} h^{-m}.
  \end{align*}
\end{proposition}
\begin{proof}
  Write \(a_k = \mathbf D^k_d\).
  Since \(\expect w_k = 0\), the sum equals \(\sum_{k=1}^N a_k w_k\).
  By Fact~\ref{fact:subgaussian-sum},
  \[
  \left\|\sum_{k=1}^N a_k w_k\right\|_{\psi_2}^2 \leq C \sum_{k=1}^N |a_k|^2 \|w_k\|_{\psi_2}^2 \leq C \nu^2 \sum_{k=1}^N a_k^2 = C \nu^2 \|\mathbf D_{\cdot, d}\|_2^2.
  \]
  Finally, \(\|\mathbf D_{\cdot, d}\|_2 \le \|\mathbf D\|\) and Proposition~\ref{prop:local-polynomial-filter} gives
  \(\|\mathbf D\| \le C(p, i_0) N^{-m - 1/2} h^{-m}\), which completes the proof.
\end{proof}

%% file: iv-sections-l4dc/cor.tex
\section{Proof of Corollary~\ref{cor:consistency}}
      The latent model equation \(Y^\star = X^\star \theta^\star\) holds if the hatted operators \(\hat {\mathcal{H}}\), \(\hat {\mathcal G}\) are replaced by the true operators.
      and \(\theta^\star = \theta_0\).
The \(n\) of the theorem should be replaced with the number of filter windows \(\sim n\).
The \(N\) of the theorem applies to the \(N\) of the estimator.
Now we turn to verifying the data bounds. For \(\bar \nu_{zy}\),
\begin{align*}
      \left\|\expect z_i y_i^\intercal - z_i (y_i^\star)^\intercal\right\|
      &= 
      \left\|\expect z_i \expect \del{y_i^\intercal - (y_i^\star)^\intercal}\right\|
      \tag{independence of \(z_i\) and \(y_i\)}
      \\
      &\leq
      \left\|\expect z_i\right\|
      \left\|\expect \del{y_i^\intercal - (y_i^\star)^\intercal}\right\|
      \\
      &\leq
      \mu
      \left\|\expect \del{y_i^\intercal - (y_i^\star)^\intercal}\right\|
      \tag{boundedness of \(z_i\) \S\ref{subsec:filters-why-mu}}
      \\
      &\leq
      C \mu
      (Nh)^{p - d}.
      \tag{bias of Lemmas~\ref{lem:local-polynomial-filter}, \ref{lem:filters-independence}}
      \intertext{where \(d=1\) in continuous-time and \(d=0\) in discrete-time. For \(\bar \nu_{zx}\), the same reasoning yields}
      \left\|\expect z_i x_i^\intercal - z_i (x_i^\star)^\intercal\right\|
      &=
      \left\|\expect z_i \expect \del{x_i^\intercal - (x_i^\star)^\intercal}\right\|
      \tag{independence of \(z_i\) and \(x_i\)}
      \\
      &\leq
      \left\|\expect z_i\right\|
      \left\|\expect \del{x_i^\intercal - (x_i^\star)^\intercal}\right\|
      \\
      &\leq
      \mu
      \left\|\expect \del{x_i^\intercal - (x_i^\star)^\intercal}\right\|
      \tag{boundedness of \(z_i\) \S\ref{subsec:filters-why-mu}}
      \\
      &\leq
      C \mu
      \sbr{(Nh)^{p} + N^{-1/2}}.
      \tag{bias and subgaussian norm of Lemmas~\ref{lem:local-polynomial-filter}, \ref{lem:filters-independence}, propagated through Lipschitz \(\phi\), Assumption~\ref{assum:phi-lip}}
      \intertext{For \(\tilde \nu_{zy}\),}
      \left\|\expect z_i y_i^\intercal - z_i y_i^\intercal\right\|_{\psi_2}
      &\leq
      \left\|z_i y_i^\intercal\right\|_{\psi_2}
      \tag{centering}
      \\
      &\leq \left\|z_i\right\|_{\psi_2}
      \left\|y_i^\intercal\right\|_{\psi_2}
      = 
      \left\|z_i\right\|_{\psi_2}
      \left\|y_i^\star + \del{\expect y_i - y_i^\star} + \del{\expect y_i - \expect y_i}\right\|_{\psi_2}
      \\
      &\leq C \mu
      \del{|y_i^\star| + (Nh)^{p - d} + N^{-d -1/2} h^{-d}}
      \tag{subgaussian norms Lemmas~\ref{lem:local-polynomial-filter}, \ref{lem:filters-independence}}.
\end{align*}
where \(d=1\) in continuous-time and \(d=0\) in discrete-time.
The same bound applies to \(\tilde \nu_{zx}\) after taking \(d=0\).
Balancing the two terms,
\begin{align*}
      (Nh)^{p - d} = N^{-d - 1/2} h^{-d}
      \implies
      N = h^{-2p/(2p + 1)}.
\end{align*}
Thus the ideal scaling is
\begin{align*}
      \bar \nu_{zy} &\lesssim \mu h^{(p-d)/(2p + 1)}
      \\
      \bar \nu_{zx} &\lesssim \mu h^{p/(2p + 1)}
      \\
      \tilde \nu_{zy} &\lesssim \mu \del{\sup_i |y_i^\star| + h^{(p-d)/(2p + 1)}}
      \\
      \tilde \nu_{zx} &\lesssim \mu \del{\sup_i |x_i^\star| + h^{p/(2p + 1)}}
\end{align*}
The function \(\gamma = \gamma_\text{head} + \gamma_\text{body} + \gamma_\text{tail}\) becomes:
\begin{align*}
      \gamma_{\text{head}}
      &= \frac{2}{\sigma^2 - \lambda} \lesssim \frac{1}{n - \lambda}
      \\
      \intertext{and}
      \gamma_{\text{tail}}(r; a, b, K)
      &=
      \frac{1}{\lambda}
      \exp\sbr{
            -\frac{
                  C (\sigma^2)^2
            }{
                  r \del{\sqrt{n h^{-2p/(2p + 1)}} \mu \sup_i |y_i^\star|}^2
            }
      }
      \\
      &\lesssim
      \frac{1}{\lambda} \exp\del{-C n^3 h^{2p/(2p + 1)}},
\end{align*}
with \(\gamma_\text{body}\) having a subdominant contribution.
Thus the condition for \(\gamma_\text{head}\) to dominate is that for all \(C\),
\begin{align*}
      \exp\del{-C n^3 h^{2p/(2p + 1)}}
      \ll \lambda
      \ll n.
\end{align*}
This also renders the first term in the conclusion of Theorem~\ref{thm:consistency} to be of subdominant order, resulting in the scaling in \(n\) and \(h\),
\begin{align*}
      \left\|\hat\theta - \theta^*\right\|_q
      &\lesssim
      \del{\bar \nu_{zx} + \bar \nu_{zy}}
      +
      \sqrt{\frac{N}{n}} \del{
            \tilde \nu_{zx} + \tilde \nu_{zy}
      }
      \\
      &\lesssim
      h^{(p-d)/(2p + 1)} + \sqrt{
            \frac{1}{
                  n h^{2p/(2p + 1)}
            }
      }.
\end{align*}

%% file: iv-sections-l4dc/lorenz-supplement.tex
\section{Supplement to \S\ref{sec:lorenz}}
\label{section:lorenz-supplement}
\subsection{Data-generating process}
We are given \(n\) measurements of \(\xi\) at \(\{ih\}_{i = 1}^n\), with i.i.d.~Gaussian noise of mean zero and variance \(\eta\):
\begin{align*}
      z_i &= \xi(ih) + \mathcal{N}(0, \eta)
\end{align*}
The number of measurements is \(n = 100000\); the sampling period is \(h = 0.001\).

The true parameter is given by
\begin{align*}
      \theta_0
      &=
      \begin{pmatrix}
            0 & 0 & 1 \\
            -10 & 28 & 0 \\
            10 & -1 & 0 \\
            0 & 0 & -8/3 \\
            0 & 0 & 1 \\
            0 & -1 & 0
      \end{pmatrix},
\end{align*}
and the initial condition is \(\xi(0) = (-8, 8, 27)\).

In the input signal, the excitation frequency is \(f = 1\).

\subsection{Continuous-time estimator}
We approximate these operators using local polynomial regression with a filter window size \(N=100\) and accuracy order \(p=75\).
The measurement noise variance is \(\eta = 0.1\).
The estimator hyperparameters are \(\lambda = 10\) and \(\mu= 200\).

\subsection{Discrete-time estimator}
We approximate these operators using local polynomial regression with a filter window size \(N=100\) and accuracy order \(p=75\).
The measurement noise variance is \(\eta = 1\).
The estimator hyperparameters are \(\lambda = 10\) and \(\mu= 200\).

\subsection{Reporting}
All values are normalized by the Frobenius norm of the (pseudo-) true parameter.
The bias is computed as the Frobenius distances between the mean of the estimator and the (pseudo-) true parameter.
The standard deviation is computed as the quadratic mean of the Frobenius distance between the estimator and its mean.
The root mean square error is computed as the quadratic mean of the Frobenius distance between the estimator and the (pseudo-) true parameter.

We run 2000 Monte Carlo trials for each estimator and compute standard errors by bootstrapping.

\subsection{Kernel density plots}
\begin{figure*}
      \includegraphics[width=\linewidth]{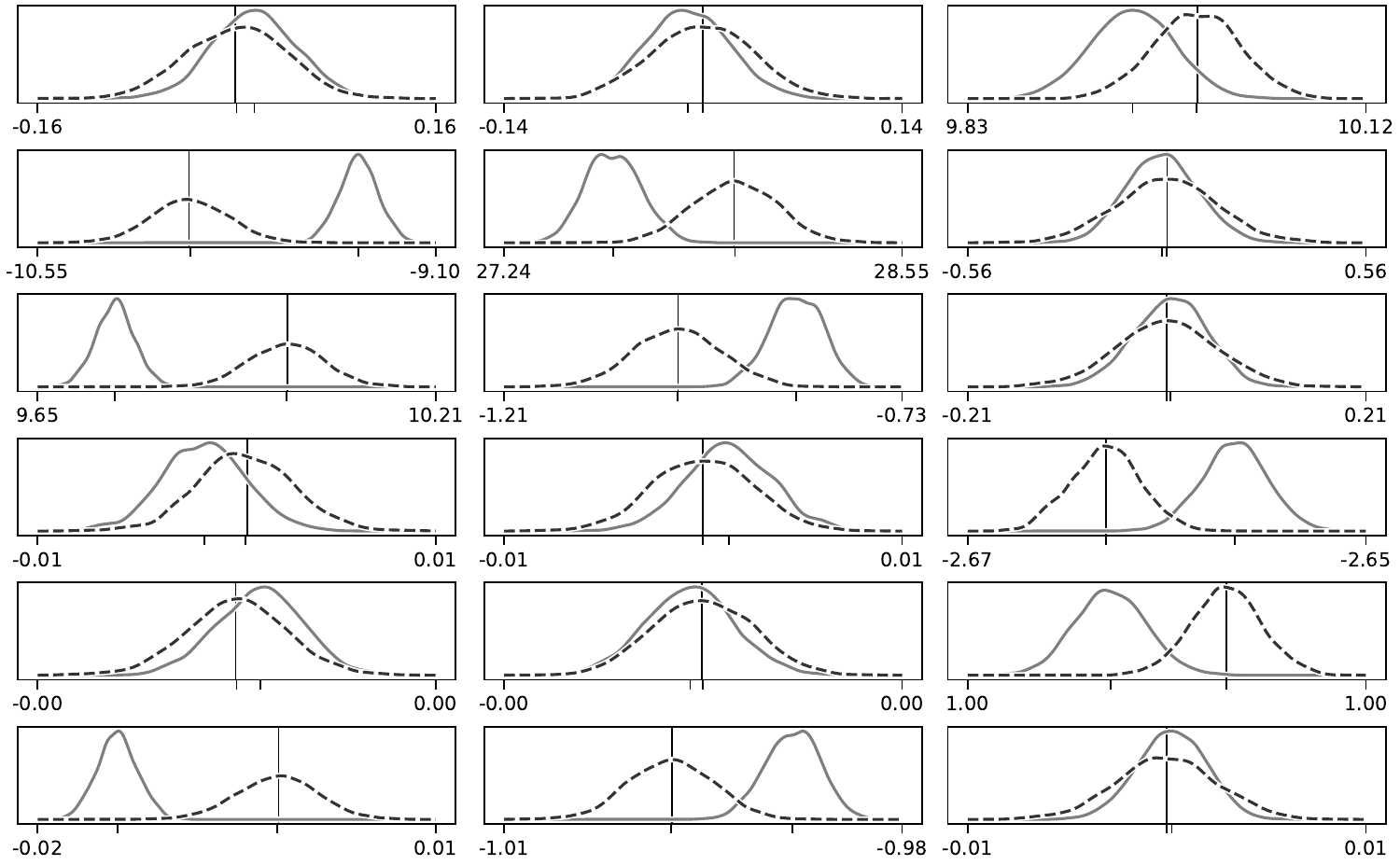}
      \caption{
            \label{fig:lorenz-continuous-distribution}
            Continuous-time Lorenz: elementwise marginal kernel density estimates of the sampling distributions
            of our estimator (dashed) and a baseline estimator (solid).
            Vertical line indicates ground truth; ticks indicate mean of sampling distribution.
      }
\end{figure*}

\begin{figure*}
      \includegraphics[width=\linewidth]{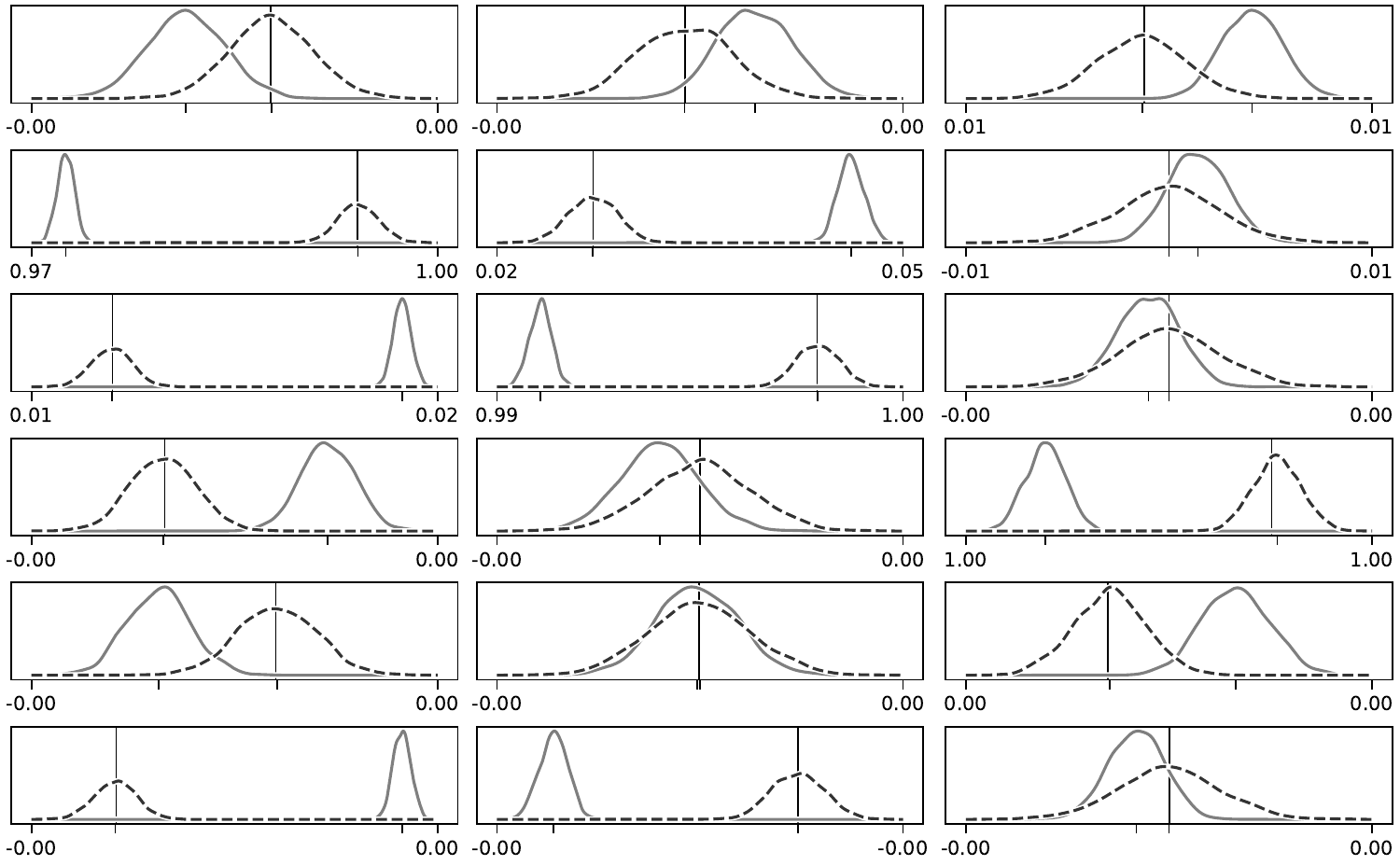}
      \caption{
            \label{fig:lorenz-discrete-distribution}
            Discrete-time Lorenz:
            Elementwise marginal kernel density estimates of the sampling distributions
            of our estimator  (dashed) and a baseline estimator  (solid).
            Vertical line indicates pseudo-true value; ticks indicate mean of sampling distribution.
      }
\end{figure*}

%% file: iv-sections-l4dc/vanderpol-supplement.tex
\section{Supplement to \S\ref{subsec:vanderpol}}
\label{section:vanderpol-supplement}
\subsection{Data-generating process}
We observe \(n\) scalar measurements of \(x(t)\) at \(\{ih\}_{i=1}^n\), with i.i.d.~Gaussian noise of mean zero and variance \(\eta\):
\begin{align*}
      z_i &= x(ih) + \mathcal{N}(0, \eta)
\end{align*}
The number of measurements is \(n = 100000\), the sampling period is \(h = 0.001\).
The Van der Pol parameter is \(\mu = 2\) and the initial condition is \((x(0), \dot x(0)) = (2, 0)\).

\subsection{Continuous-time estimator}
We approximate \(\mathcal H\) and \(\mathcal G\) using local polynomial regression with a filter window size \(N = 100\) and accuracy order \(p = 20\).
The stencil outputs derivatives of order \(d = 0, 1, 2\): rows \(d = 0, 1\) implement \(\hat{\mathcal G}\) (or \(\tilde{\mathcal G}\)) and row \(d = 2\) implements \(\hat{\mathcal H}\).
The measurement noise variance is \(\eta = 10^{-4}\).
The estimator hyperparameters are \(\lambda = 1\) and \(\mu = 200\).

\subsection{Reporting}
All values are normalized by the Euclidean norm of the true parameter \(\theta_0\).
The bias is computed as the Euclidean distance between the mean of the estimator and \(\theta_0\).
The standard deviation is computed as the quadratic mean of the Euclidean distance between the estimator and its mean.
The root mean square error is computed as the quadratic mean of the Euclidean distance between the estimator and \(\theta_0\).

We run 2000 Monte Carlo trials for each estimator and compute standard errors by bootstrapping.

\begin{table}
      \centering
      \input{iv-koopman/output/vanderpol-continuous-error.tex}
      \caption{\label{tab:vanderpol-continuous-error}%
      Comparison of bias, standard deviation, and root mean square error of our estimator and a least squares estimator for the parameter of the Van der Pol oscillator.
      Monte Carlo standard errors in parentheses.
      }
\end{table}

%% file: iv-koopman/output/vanderpol-continuous-error.tex
\begin{tabular}{lrrr}
\toprule
Estimator & abs. bias (\%) & std (\%) & rmse (\%) \\
\midrule
Instrumental Variables (ours) & \num{0.489 \pm 0.258} & \num{13.205 \pm 0.211} & \num{13.214 \pm 0.211} \\
Least Squares & \num{17.533 \pm 0.018} & \num{1.622 \pm 0.022} & \num{17.608 \pm 0.018} \\
\bottomrule
\end{tabular}